\documentclass{article}

\usepackage{fullpage}

\usepackage[all=normal,bibliography=tight,bibnotes=tight]{savetrees}

\usepackage{tikz}
\usetikzlibrary{snakes}
\usepackage{xspace}
\usepackage{comment}

\usepackage{amsthm,amsmath,amstext,amssymb}

\begin{document}

\newcommand{\twoconn}{{\sc{2-Disjoint Connected Subgraphs}}\xspace}
\newcommand{\lce}{{\sc{Line Cluster Embedding}}\xspace}
\newcommand{\adp}{{\sc{Acyclic Digraph Partition}}\xspace}
\newcommand{\setsplitting}{{\sc{Set Splitting}}\xspace}
\newcommand{\colourhittingset}{{\sc{Colour Hitting Set}}\xspace}
\newcommand{\tcnfsat}{$3${\sc{-CNF-SAT}}\xspace}
\newcommand{\sat}{{\sc{SAT}}\xspace}
\newcommand{\czas}[1]{t(#1)}
\newcommand{\koszt}[1]{T(#1)}
\newcommand{\tudu}[1]{{\bf{TODO: #1}}}
\newcommand{\eps}{\varepsilon}
\newcommand{\pinezka}{\textrm{nil}}
\newcommand{\pred}{{\ensuremath{pred}}}
\newcommand{\inn}{{{\rm in}}}
\newcommand{\out}{{{\rm out}}}
\newcommand{\succc}{{\ensuremath{succ}}}
\newcommand{\ff}{\varphi}
\newcommand{\Ii}{{\ensuremath{\mathcal{I}}}}
\newcommand{\cF}{{\ensuremath{\mathcal{F}}}}
\newcommand{\cS}{{\ensuremath{\mathcal{S}}}}
\newcommand{\N}{{\ensuremath{\mathbb{N}}}}
\newcommand{\R}{{\ensuremath{\mathbb{R}}}}
\renewcommand{\subset}{{\subseteq}}
\newcommand{\unlesscompass}{\ensuremath{\textrm{NP} \subseteq \textrm{coNP}/\textrm{poly}}}
\newcommand{\Ohstar}{O^\star}
\newcommand{\Thraves}{Thraves}

\newtheorem{theorem}{Theorem}[section]
\newtheorem{lemma}[theorem]{Lemma}
\newtheorem{definition}[theorem]{Definition}
\newtheorem{corollary}[theorem]{Corollary}
\newtheorem{remark}[theorem]{Remark}
\newtheorem{proposition}[theorem]{Proposition}

\newcommand{\defproblemnoparam}[3]{
  \vspace{1mm}
\noindent\fbox{
  \begin{minipage}{\textwidth}  
  #1 \\ 
  {\bf{Input:}} #2  \\
  {\bf{Task:}} #3
  \end{minipage}
  }
  \vspace{1mm}
}

\author{Marek Cygan\thanks{Institute of Informatics, University of Warsaw, Poland, \texttt{cygan@mimuw.edu.pl}} \and
	Marcin Pilipczuk\thanks{Institute of Informatics, University of Warsaw, Poland, \texttt{malcin@mimuw.edu.pl}} \and
	Micha\l{} Pilipczuk\thanks{Department of Informatics, University of Bergen, Norway, \texttt{michal.pilipczuk@ii.uib.no}} \and
	Jakub Onufry Wojtaszczyk\thanks{Google Inc., Warsaw, Poland, \texttt{onufry@google.com}}}
\date{}

\title{Sitting closer to friends than enemies, revisited}

\maketitle

\begin{abstract}
Signed graphs, i.e., undirected graphs with edges labelled with a plus or minus sign,
are commonly used to model relationships in social networks.
Recently, Kermarrec and Thraves~\cite{francuzi} initiated the study
of the problem of appropriately visualising the network:
They asked whether any signed graph can be embedded into the metric space $\R^l$
in such a manner that every vertex is closer to all its friends (neighbours via positive edges)
than to all its enemies (neighbours via negative edges).
Interestingly, embeddability into $\R^1$ can be expressed as a purely combinatorial problem. In this paper we pursue a deeper study of this particular case, answering several questions posed by Kermarrec and Thraves.

First, we refine the approach of Kermarrec and Thraves for the case of complete signed graphs
by showing that the problem is closely related to the recognition of proper interval graphs.
Second, we prove that the general case, whose polynomial-time tractability remained open,
is in fact $NP$-complete.
Finally, we provide lower and upper bounds for the time complexity of the general case:
we prove that the existence of a subexponential time
(in the number of vertices and edges of the input signed graph)
algorithm would violate the Exponential Time Hypothesis,
whereas a simple dynamic programming approach gives a running time single-exponential in the number of vertices.
\end{abstract}

\section{Introduction}

Undirected graphs with edges labelled positively (by a~$+$)
and negatively (by a~$-$), called {\em signed graphs},
in many applications serve as a very simple model of
relationships between a group of people, e.g., in a social network.
Sign labels can express in a simplified way mutual relations,
like staying in a relationship, family bonds or conflicts, by classifying them
either as {\em friendship} ($+$ edge), {\em hostility} ($-$ edge)
or {\em ambivalence} (no edge).
In particular, much effort has been put into properly understanding and
representing the structure of the network, balancing it or naturally partitioning
into clusters \cite{signed0,signed1,signed2,signed2ipol,signed3,signed4,signed5,signed6}.
One of the problems is to visualize the model graph properly, i.e.,
in such a way that positive relations tend to make vertices be placed close to each other,
while negative relations imply large distances between vertices.

In their recent work, Kermarrec and Thraves~\cite{francuzi} formalized this
problem as follows:
Consider the metric space $\R^l$ with the Euclidean metric denoted by $d$.
Given a signed graph $G$, is it possible to embed the vertices of $G$
in $\R^l$ so that for any positive edge $uu_1$ and negative edge $uu_2$
it holds that $d(u, u_1) < d(u, u_2)$?
This question has a natural interpretation:
we would like to place a group of people so that every person
is placed closer to his friends than to his enemies. 

The work of Kermarrec and Thraves \cite{francuzi} concentrated on showing
a number of examples and counterexamples for embeddability into spaces of small
dimensions ($1$ and $2$) and a deeper study of the $1$-dimensional case.
Interestingly enough, the case of the Euclidean line has an equivalent
formulation in the language of pure combinatorics:
Given a signed graph $G$, is it possible to order the vertices of $G$
so that for any positive edge $uw$
there is no negative edge $uv$ with $v$ laying between $u$ and $w$?
The authors made algorithmic use of this combinatorial insight:
Providing the given signed graph is complete
(i.e., every pair of vertices is adjacent via a positive or negative edge)
they show a polynomial-time algorithm that computes an embedding into a line
or reports that no such embedding exists.

Kermarrec and Thraves also posed a number of open problems in the area,
including the question of the complexity of determining the embeddability
of an arbitrary (not necessarily complete) graph into the Euclidean line.

\paragraph{Our results.}
In this paper, we focus on the problem of embedding a signed graph into a line. The reformulation of the $1$-dimensional case, proven by Kermarrec and Thraves, turns out to be an interesting combinatorial problem, which allows classical methods of analysis and shows interesting links with the class of proper interval graphs.
 
We begin with refining the result of Kermarrec and Thraves for the case of complete graphs.
We prove that a complete signed graph is embeddable into a line
if and only if the graph formed by the positive edges is a proper interval graph.
Using this theorem one can immediately transfer all the results
from the well-studied area of proper interval graphs into our setting.
Most importantly, as recognition of proper interval graphs can be performed in
linear-time~\cite{pr-int-recognition},
we obtain a simpler algorithm for determining the embedability of a complete
graph into a line, with a linear runtime.

We next analyse the general case.
We resolve the open problem posed in~\cite{francuzi} negatively:
it is $NP$-complete to resolve whether a given signed graph can be embedded into a line.
This hardness result also answers other questions of Kermarrec and Thraves \cite{francuzi}.
For example, we infer that it is $NP$-hard to decide the smallest dimension
of a Euclidean space in which the graph can be embedded, as such an algorithm could be used to test embeddability into a line.

Furthermore, we are able to show a lower bound on the time complexity
of resolving embeddability into a line, 
under a plausible complexity assumption.
We prove that obtaining an algorithm running in subexponential time
(in terms of the total number of vertices and edges of the input graph)
would contradict the {\em Exponential Time Hypothesis} \cite{seth}
(see Section \ref{sec:prel} for an exact statement).
We complete the picture of the complexity of the problem by showing
a dynamic programming algorithm that runs in $\Ohstar(2^n)$ time
\footnote{The $\Ohstar()$ notation surpresses factors
that are polynomial in the input size.},
matching the aforementioned lower bound up to a constant
in the base of the exponent
($n$ denotes the number of vertices of the input graph).

\paragraph{Organisation of the paper.}
In Section~\ref{sec:prel} we recall widely known notions and facts
that are of further use,
and provide the details of the combinatorial reformulation of the problem
by Kermarrec and Thraves \cite{francuzi}.
Section~\ref{sec:comp} is devoted to refinements in the analysis
of the case of the complete signed graphs,
while Section~\ref{sec:gen} describes upper and lower bounds
for the complexity of the general case.
Finally, in Section~\ref{sec:conc} we gather conclusions and ideas for further work.

\section{Preliminaries}\label{sec:prel}

\paragraph{Basic definitions.}
For a finite set $V$, by an {\em{ordering}} of $V$ we mean a bijection
$\pi: V \to \{1,2,\ldots,|V|\}$.
We sometimes treat an ordering $\pi$ as a linear order on $V$
and for $u,v \in V$ we write $u \leq_\pi v$ to denote $\pi(u) \leq \pi(v)$.

In a graph $G=(V,E)$ the {\em neighbourhood} of a vertex $v$,
denoted $N(v)$, is the set of all its neighbours,
i.e., $\{w: vw\in E\}$.
The {\em closed neighbourhood} of $v$ is defined as $N[v]=N(v)\cup \{v\}$.

A {\em signed graph} is a triple $G=(V,E^+,E^-)$,
where $E^+,E^-\subseteq V^{[2]}$ and $E^+\cap E^-=\emptyset$.
We view a signed graph as an undirected simple graph
with two possible labels on the edges:
positive ($+$) and negative ($-$).
We call the edges from $E^+$ positive, while those from $E^-$ --- negative.
The graph $G^+=(V,E^+)$ is called the {\em positive part} of $G$,
and $G^-=(V,E^-)$ --- the {\em negative part}.
A signed graph is called {\em complete} if $E^+\cup E^-=V^{[2]}$,
i.e., every pair of vertices is adjacent via a positive or negative edge.

\paragraph{Proper interval graphs.}
Let $G=(V,E)$ be an undirected graph, $\Ii$ be a family of size $|V|$
of intervals on real line with nonempty interiors
and pairwise different endpoints and $\iota:V\to \Ii$ be any bijection.
We say that $\Ii$ is an {\em interval model} for $G$ if for every $v,w\in V$,
$v\neq w$, $vw\in E$ is equivalent to $\iota(v)\cap \iota(w)\neq \emptyset$.
$\Ii$ is a {\em proper interval model} if, additionally,
none of the intervals is entirely contained in any other.
Graphs having an interval model are called {\em interval graphs},
while if a proper interval model exists as well,
we call them {\em proper interval graphs}.
We will omit the mapping $\iota$ whenever it is clear from the context.

\paragraph{Exponential Time Hypothesis} \cite{seth}:
The {\em Exponential Time Hypothesis} (ETH for short)
asserts that there exists a constant $C > 0$
such that no algorithm solving the \tcnfsat{} problem in $O(2^{Cn})$ exists,
where $n$ denotes the number of variables in the input formula.

\paragraph{Combinatorial problem statement.}
In~\cite{francuzi}, Kermarrec and Thraves work with the metric definition of the problem:
Given a signed graph $G=(V,E^+,E^-)$ a feasible embedding of $G$ in the Euclidean space $\R^l$
is such a function $f:V\to \R^l$ that for all $u_1,u_2,u$, if $u_1u\in E^+$ and $u_2u\in E^-$,
then $d(f(u_1),f(u))<d(f(u_2),f(u))$
(recall that $d$ stands for the Euclidean distance in $\R^l$).
However, for the $1$-dimensional case they have in essence proved the following result:
\begin{theorem}[Lemmata $3$ and $4$ of~\cite{francuzi}, rephrased]
A signed graph $G=(V,E^+,E^-)$ has a feasible embedding in a line
iff there is an ordering $\pi$ of $V$ such that for every $u\in V$:
\begin{itemize}
\item[{\em (i)}] there are no $u_1<_\pi u_2<_\pi u$ such that $u_1u\in E^+$ and $u_2u\in E^-$;
\item[{\em (ii)}] there are no $u_1>_\pi u_2>_\pi u$ such that $u_1u\in E^+$ and $u_2u\in E^-$.
\end{itemize}
\end{theorem}
We will jointly call conditions {\em (i)} and {\em (ii)} {\em the condition imposed on $u$}.
Somewhat abusing the notation,
the ordering $\pi$ will also be called {\em an embedding of $G$ into the line}.
Therefore, from now on we are working with the following combinatorial problem
that is equivalent to the version considered by Kermarrec and Thraves:

\defproblemnoparam{\lce}
{A signed graph $G=(V,E^+,E^-)$.}
{Does there exist an ordering $\pi$ on $V$ such that for every $u\in V$:
{\em (i)} there are no $u_1<_\pi u_2<_\pi u$ such that $u_1u\in E^+$ and $u_2u\in E^-$;
{\em (ii)} there are no $u_1>_\pi u_2>_\pi u$ such that $u_1u\in E^+$ and $u_2u\in E^-$.}

\newcommand{\rla}{\leftarrow}
\newcommand{\lra}{\rightarrow}
\newcommand{\pos}[1]{\pi(#1)}

\section{The complete signed graph case}\label{sec:comp}

In their work, Kermarrec and Thraves~\cite{francuzi} announced
a polynomial-time algorithm solving the \lce problem in the case
where the input signed graph is complete.
Their line of reasoning was essentially as follows:
if a signed graph can be embedded into a line,
then its positive part has to be chordal.
However, for a connected chordal graph with at least $4$ vertices
that actually is embeddable into a line,
every {\em perfect elimination ordering} of the graph
is a feasible solution.
Therefore, having checked that the graph is chordal
and computed a perfect elimination ordering of every connected component,
we can simply verify whether the obtained ordering is a correct line embedding.

We refine the approach of Kermarrec and Thraves by showing
that a complete graph has a line embedding if and only if
its positive part is a proper interval graph.
Recall that proper interval graphs are a subclass of chordal graphs;
therefore, the result nicely fits into the picture of Kermarrec and Thraves.
Moreover, the theory of proper interval graphs is well-studied,
so many results from that area can be immediately translated to our setting.
For instance, many NP-complete problems become solvable in polynomial time
on proper interval graphs (e.g., \cite{martin,proper:hamilton,proper:longestpath,proper:kcluster}),
and the linear-time algorithm of Corneil et al.~\cite{pr-int-recognition}
for recognizing proper interval graphs immediately solves the \lce problem in
linear time in case of a complete signed graph.

\begin{theorem}\label{thm:pr-int}
A complete signed graph $G=(V,E^+,E^-)$ is embeddable in $\R^1$ if and only if
$G^+=(V,E^+)$ is a proper interval graph.
Moreover, having a feasible ordering $\pi$ of vertices of $V$,
a proper interval model of $G^+$ sorted with respect to the left ends of the intervals
can be computed in linear time;
conversely, having a proper interval model of $G^+$ sorted with respect
to the left ends of the intervals,
we can compute a feasible ordering $\pi$ in linear time.
\end{theorem}
\begin{proof}
First, let us assume that $G^+$ is a proper interval graph,
and let $\Ii=\{I_v:v\in V\}$ be a proper interval model of $G^+$.
Notice that as no interval is contained in another, we have a natural
order on $\Ii$ --- ordering the intervals with respect to the left
ends (or, equivalently, the right ends).
We claim that $\pi$ is a feasible solution for the \lce instance $G=(V,E^+,E^-)$.
Take any $u\in V$.
Assume that there were some $u_1<_\pi u_2<_\pi u$ such that $u_1u\in E^+$ and $u_2u\in E^-$;
this implies intervals $I_{u_1}$ and $I_{u}$ would overlap.
This, in turn, means that the right end of interval $I_{u_1}$
would be on the right of the left end of interval $I_u$.
Therefore, the left and right ends of $I_{u_2}$ are on different sides of the left end of $I_u$,
as $u_2<_\pi u$ and $u_1<_\pi u_2$, so $I_{u_2}$ and $I_u$ overlap.
This is a contradiction with $u_2u\notin E^+$.
A symmetrical argument for the second case finishes the proof in this direction.

Now let us assume that $G$ is embeddable in the line and let $\pi$,
an ordering of $V$, be a solution.
Moreover, let $v^\rla$ be the first (with respect to $\pi$) vertex in the closed neighbourhood of $v$ in $G^+$,
while let $v^\lra$ be the last.
Of course, $v^\rla\leq_\pi v\leq_\pi v^\lra$.
Let us define a family of intervals on $\R$:
let $I_v=\left[\pos{v},\pos{v^\lra}+\frac{\pos{v}}{|V|+1}\right]$ for $v\in V$
and $\Ii=\{I_v:v\in V\}$.
Observe that intervals $I_v$ have nonempty interior and pairwise different endpoints.
Now, we prove that (1) no $I_v$ is fully contained in some other $I_w$ and
(2) for all $v,w\in V$, $vw\in E^+$ if and only if $I_v\cap I_w\neq \emptyset$.
This suffices to show that $\Ii$ is a proper interval model for $G^+$.

In order to establish (1), let us assume the contrary:
there exists a pair of vertices $v,w$ such that $\pos{v}>\pos{w}$
and $\pos{v^\lra}+\frac{\pos{v}}{|V|+1}<\pos{w^\lra}+\frac{\pos{w}}{|V|+1}$.
Then $\pos{v^\lra}<\pos{w^\lra}$.
Therefore, by definition of $v^\lra$, $vw^\lra\in E^-$.
On the other hand, $ww^\lra\in E^+$ and $w<_\pi v<_\pi w^\lra$,
a contradiction with the assumption that $\pi$ was a proper embedding.

Now we proceed to the proof of (2).
Take any two distinct vertices $v,w$, without losing generality assume that $v<_\pi w$.
If $vw\in E^+$, then $\pos{v}<\pos{w}$
and $\pos{w}\leq \pos{v^\lra}<\pos{v^\lra}+\frac{\pos{v}}{|V|+1}$,
so $I_v$ and $I_w$ overlap.
On the other hand if $vw\in E^-$, then from the condition imposed on $v$ it follows that $w>_\pi v^\lra$.
Consequently, $\pos{w}>\pos{v^\lra}$ and, as $\pos{v}<|V|+1$,
also $\pos{w}>\pos{v^\lra}+\frac{\pos{v}}{|V|+1}$.
Therefore, in this case $I_v$ and $I_w$ do not overlap.
This proves $\Ii$ is in fact a proper interval model for $G^+$.

The algorithmic part of the theorem statement
follows directly from the presented constructions. 
\end{proof}

Let us recall the result of Corneil et al.~\cite{pr-int-recognition},
which states that proper interval graphs can be recognized in linear time
and the algorithm can also output an ordering of the vertices with 
respect to the left ends of intervals in some model.
We can pipeline this routine with Theorem~\ref{thm:pr-int} in order to obtain the following corollary:

\begin{theorem}
Assuming the input graph is complete and given as the set of positive edges,
\lce can be solved in $O(|V|+|E^+|)$ time complexity.
Moreover, the algorithm can produce a feasible ordering
of the vertices in the same time, if such an ordering exists.
\end{theorem}

\section{The general case}\label{sec:gen}

\subsection{$NP$-completeness of the general case}

In \cite{francuzi} Kermarrec and Thraves asked whether the \lce problem
is also polynomial-time solvable in the case where the input is not restricted to complete graphs.
In this section we show that this is unlikely: in fact, the problem becomes $NP$-complete.
The proof consists of two steps.
First, using a reduction from the \setsplitting problem we show
that an auxiliary problem, called \adp, is $NP$-complete.
Next, we reduce \adp to \lce.
We believe that the \adp can turn out to be a useful pivot problem also in other hardness reductions.

\defproblemnoparam{\adp}{A directed graph $D=(V,A)$.}{
Is it possible to partition $V$ into two sets
$V_1$ and $V_2$, so that both $D[V_1]$ and $D[V_2]$ are directed acyclic graphs (DAGs)?}

\noindent Let us also recall the definition of the NP-complete \setsplitting problem \cite{garey-johnson}.

\defproblemnoparam{\setsplitting}{A set system $(\cF,U)$, where $\cF \subseteq 2^U$.}{
Does there exist a subset $X \subseteq U$
such that each set in $\cF$ contains both an element from $X$
and an element from $U\setminus X$?}

\begin{lemma}\label{lem:toadp}
There exists a polynomial-time algorithm
that given an instance $(\cF,U)$ of \setsplitting
outputs an equivalent instance $G=(V,A)$ of \adp,
for which $|V|=|U|+\sum_{F\in \cF} |F|$ and $|A|=3\sum_{F\in \cF} |F|$.
\end{lemma}
\begin{proof}
We construct the directed graph $D=(V,A)$ as follows.
For every set $F\in \cF$ and every $u\in F$ we build a vertex $c_u^F$
and connect all the vertices corresponding to the same set $F$
into a directed cycle in any order.
For every element $u\in U$ we build a vertex $d_u$ and for every vertex of the form $c_u^F$
we introduce two arcs: $(d_u,c_u^F)$ and $(c_u^F,d_u)$.
This concludes the construction; it is easy to verify the claimed sizes of $V$ and $A$.

Let us formally prove that the instances are equivalent.
Let $X$ be any solution to the $(\cF,U)$ instance of \setsplitting.
Let $V_1=\{d_u:u\in X\}\cup \{c_u^F:u\in U\setminus X\}$
and $V_2=\{d_u:u\in U\setminus X\}\cup \{c_u^F:u\in X\}$.
As $X$ splits every set $F\in \cF$, none of the cycles formed by vertices $c_u^F$
for fixed $F$ is entirely contained in either $V_1$ or $V_2$.
Also, for every element $u$ the vertex $d_u$ becomes isolated
in the corresponding graph $D[V_i]$,
as all his neighbours belong to $V_{3-i}$.
Therefore, both $D[V_1]$ and $D[V_2]$ are collections of isolated vertices and directed paths
and $(V_1,V_2)$ is a solution to the \adp instance.

In the other direction, let $(V_1,V_2)$ be a solution to the instance of \adp.
Let $X=\{u:d_u\in V_1\}\subseteq U$,
we claim that $X$ is a solution to the instance of \setsplitting.
Take any $F\in \cF$.
As the cycle formed by vertices $c_u^F$ is not entirely contained
in any of the graphs $D[V_1], D[V_2]$, there exist
some $u_1$ such that $c_{u_1}^F\in V_1$ and $u_2$ such that $c_{u_2}^F\in V_2$.
As the cycles formed by pairs $\{d_{u_1},c_{u_1}^F\}$ and $\{d_{u_2},c_{u_2}^F\}$
are also not entirely contained in $D[V_1]$ nor in $D[V_2]$,
$d_{u_1}\in V_2$ and $d_{u_2}\in V_1$.
Consequently, $u_1\in U\setminus X$ and $u_2\in X$ and the set $F$ is split. 
\end{proof}

\begin{lemma}\label{lem:fromadp}
There exists a polynomial-time algorithm that given an instance $D=(V,A)$ of \adp
outputs an equivalent instance $H=(V',E^+,E^-)$ of \lce,
such that $|V'|=|V|+|A|+1$, $|E^+|=2|A|$ and $|E^-|=|A|+|V|$.
\end{lemma}
\begin{proof}
We construct the graph $H$ as follows: The set of vertices, $V'$, consists of:
\begin{itemize}
\item a {\em special} vertex $s$;
\item for every $e\in A$, a {\em checker} vertex $c_e$;
\item for every $v\in V$, an {\em alignment} vertex $a_v$.
\end{itemize}
We construct the edges of the signed graph as follows:
\begin{itemize}
\item for every $e\in A$, we introduce a positive edge $sc_e$;
\item for every $v\in V$, we introduce a negative edge $sa_v$;
\item for every arc $(v,w)\in A$, we introduce a positive edge $c_{(v,w)}a_v$ and a negative edge $c_{(v,w)}a_w$.
\end{itemize}
This concludes the construction; it is easy to verify the claimed sizes of $V',E^+,E^-$.

Let us now formally prove equivalence of the instances.
Let $\pi$, an ordering of $V'$, be a solution of the \lce instance $(V',E^+,E^-)$.
As the special vertex $s$ is adjacent via positive edges to all the checker vertices,
and via negative edges to all the other, alignment, vertices,
in the ordering $\pi$ the checker vertices
together with the special vertex have to form an interval.
Let $V_1$ be the set of those $v\in V$ for which $a_v$ is to the left of this interval,
whereas $V_2$ is the set of those $v\in V$ for which $a_v$ is to the right of this interval.
Formally, $V_1=\{v\in V:a_v\leq_\pi s\}$ and $V_2=\{v\in V:a_v\geq_\pi s\}$.
We claim that $(V_1,V_2)$ is a feasible solution of the \adp instance $(V,A)$.
Consider any arc $(v,w)$ such that $v,w\in V_1$.
As $a_v\leq_\pi c_{(v,w)}$, $a_w\leq_\pi c_{(v,w)}$, $c_{(v,w)}a_v\in E^+$ and $c_{(v,w)}a_w\in E^-$, then it follows that $a_w\leq_\pi a_v$.
Thus, $\pi$ has to induce a reverse topological ordering on the vertices of $D[V_1]$ and, therefore, $D[V_1]$ has to be acyclic.
Symmetrically, $D[V_2]$ has to be acyclic as well, which concludes the proof of $(V_1,V_2)$ being a feasible solution.

Now take any solution $(V_1,V_2)$ of \adp instance $(V,A)$.
Let $\pi_1$ be any topological ordering of $D[V_1]$ and $\pi_2$ be any topological ordering of $D[V_2]$,
by which we mean that if $(u,v)$ is an arc of $D[V_1]$, $\pi_1(u) < \pi_1(v)$,
and the same holds for $\pi_2$.
Let us construct an ordering $\pi$ of $V'$ as follows:
\begin{itemize}
\item first, we place all the vertices $a_v$ for $v\in V_1$ in the reverse order induced by $\pi_1$;
\item then, we place all the checker vertices $c_{(v,w)}$ for which $v\in V_1$ and $w\in V_2$, in any order;
\item then, we place all the checker vertices $c_{(v,w)}$ for which $v,w\in V_1$,
in reverse lexicographic order imposed by $\pi_1$ on pairs $(v,w)$;
\item then, we place the special vertex $s$;
\item then, we place all the checker vertices $c_{(v,w)}$ for which $v,w\in V_2$, in lexicographic order imposed by $\pi_2$ on pairs $(v,w)$;
\item then, we place all the checker vertices $c_{(v,w)}$ for which $v\in V_2$ and $w\in V_1$, in any order;
\item finally, we place all the vertices $a_v$ for $v\in V_2$ in the order induced by $\pi_2$.
\end{itemize}
We claim that such $\pi$ is a feasible solution to \lce instance $(V',E^+,E^-)$.

Note that the positive neighbours of the special vertex $s$ form an interval,
therefore the condition imposed on this vertex is satisfied.
Now consider a checker vertex $c_{(v,w)}$.
If $v,w$ belong to different sets $V_1,V_2$, then the only negative neighbour of $c_{(v,w)}$
is the first or the last of his closed neighbourhood with respect to $\pi$,
thus satisfying the condition imposed on $c_{(v,w)}$.
In case when $v,w\in V_1$ or $v,w\in V_2$ this is also true,
due to $\pi_1,\pi_2$ being topological orderings of $D[V_1]$, $D[V_2]$ respectively.

Now take any vertex $a_v$, by symmetry assume $v\in V_1$.
We need to prove that the condition imposed on $a_v$ is satisfied as well. The neighbours of $v$ consist of:
\begin{enumerate}
\item positive neighbours $c_{(v,v')}$, such that $v'\in V_2$;\label{vv:1}
\item positive neighbours $c_{(v,v')}$, such that $v'\in V_1$;\label{vv:2}
\item negative neighbours $c_{(v',v)}$, such that $v'\in V_1$;\label{vv:3}
\item negative neighbours $c_{(v',v)}$, such that $v'\in V_2$.\label{vv:4}
\end{enumerate}
We now verify that by the construction of $\pi$ the neighbours of $a_v$ lie in this very order with respect to $\pi$.
Clearly, the order in which we placed the checkers in $\pi$ ensures that the neighbours from (\ref{vv:1}) are placed before the neighbours from (\ref{vv:2}) and
that the neighbours from (\ref{vv:3}) are placed before the neighbours from (\ref{vv:4}). 
Thus the only non-trivial check is whether the vertices from (\ref{vv:2}) lie before the vertices from (\ref{vv:3}).
Assume otherwise, that there are some $v_1',v_2'$ such that $(v,v_1')\in A$, $(v_2',v)\in A$, but $c_{(v,v_1')}>_{\pi}c_{(v_2',v)}$. But then $v_2'<_{\pi_1}v$ as $\pi_1$ is a topological ordering of $D[V_1]$, so the pair $(v_2',v)$ is lexicographically smaller than the pair $(v,v_1')$. Thus $c_{(v,v_1')}>_{\pi}c_{(v_2',v)}$ is a contradiction with the construction of $\pi$.

We have verified that for all the vertices the conditions imposed on them are satisfied, so the instances are equivalent. 
\end{proof}

The $NP$-completeness of the \setsplitting problem \cite{garey-johnson}, together with Lemmata \ref{lem:toadp}, \ref{lem:fromadp} and a trivial observation that \lce is in $NP$, gives us the following theorem.

\begin{theorem}
The \lce problem is $NP$-complete.
\end{theorem}

As mentioned before, the question of finding the smallest dimension of the Euclidean space, into which the given graph can be embedded, clearly generalizes testing embeddability into a line. Therefore, we have the following corollary.

\begin{corollary}
It is $NP$-hard to decide the smallest dimension of the Euclidean space, into which a given signed graph can be embedded.
\end{corollary}

\subsection{Lower bound on the complexity}

In this subsection we observe that the presented chain of reductions in fact enables us also to establish a lower bound on the complexity of solving \lce under ETH. Firstly, let us complete the chain of the reductions.

\begin{lemma}\label{lem:ss}
There exists a polynomial-time algorithm that given an instance $\ff$ of \tcnfsat with $n$ variables and $m$ clauses, outputs an equivalent instance $(U,\cF)$ of \setsplitting with $|U|=2n+1$ and $\sum_{F\in \cF} |F|=2n+4m$.
\end{lemma}
\begin{proof}
We construct the instance $(U,\cF)$ as follows. The universe $U$ consists of one special element $s$ and two literals $x,\neg x$ for every variable $x$ of $\ff$. The family $\cF$ includes
\begin{itemize}
\item for every variable $x$, a set $F_x=\{x,\neg x\}$;
\item for every clause $C$, a set $F_C$ consisting of $s$ and all the literals in $C$.
\end{itemize}
It is easy to check the claimed sizes of $U,\cF$. We claim that the instance of \setsplitting $(U,\cF)$ is equivalent to the instance $\ff$ of \tcnfsat.

Assume that $\psi$ is a boolean evaluation of variables of $\ff$ that satisfies $\ff$. We construct a set $X\subseteq U$ as follows: $X$ consists of all the literals that are true in $\psi$. Now, every set $F_x$ is split, as exactly one of the literals is true and one is false, whereas every set $F_C$ is split as well, as it contains a true literal, which belongs to $X$, and the special element $s$, which does not.

Now assume that $X\subseteq U$ is a solution to the \setsplitting instance $(U,\cF)$. As taking $U\setminus X$ instead of $X$ also yields a solution, without losing generality we can assume that $s\notin X$. Every set $F_x$ is split by $X$; therefore, exactly one literal of every variable belongs to $X$ and exactly one does not. Let $\psi$ be a boolean evaluation of variables of $\ff$ such that it satisfies all the literals belonging to $X$. Observe that $\psi$ satisfies $\ff$: for every clause $C$ the set $F_C$ has to be split, so, as $s\notin X$, one of the literals of $C$ belongs to $X$ and, thus, is satisfied by $\psi$. 
\end{proof}

Note that by pipelining Lemmata \ref{lem:ss}, \ref{lem:toadp} and \ref{lem:fromadp}, we obtain a reduction from \tcnfsat to \lce, where the output instance has a number
of vertices and edges bounded linearly in the number of variables and clauses of the input formula.
This observation, together with the key tool used in proving complexity lower bounds under Exponential Time Hypothesis,
     namely the Sparsification Lemma \cite{sparsification}, gives us the following theorem.

\begin{theorem}
Unless ETH fails, there is a constant $\delta>0$ such that there is no algorithm that given a $(V,E^+,E^-)$ instance of \lce problem, solves it in $O(2^{\delta(|V|+|E^+|+|E^-|)})$ time.
\end{theorem}
\begin{proof}
Let us begin by recalling the {\em Sparsification Lemma}.

\begin{lemma}[Sparsification Lemma, Corollary $1$ of \cite{sparsification}]
For all $\varepsilon>0$ and positive $k$, there is a constant $C$ so that any
$k${\sc -SAT} formula $\Phi$ with $n$ variables can be expressed as $\Phi=\bigvee_{i=1}^{t} \Psi_i$, where $t\leq 2^{\varepsilon n}$
and each $\Psi_i$ is a $k${\sc -SAT} formula with at most $Cn$ clauses. Moreover, this disjunction
can be computed by an algorithm running in time $\Ohstar(2^{\varepsilon n})$.
\end{lemma}

Let us now assume that for all $\delta>0$ there exists an algorithm solving \lce in $O(2^{\delta(|V|+|E^+|+|E^-|)})$ time complexity. We now show an algorithm solving \tcnfsat in $\Ohstar(2^{\varepsilon n})$ time for every $\varepsilon>0$, where $n$ is the number of variables, thus contradicting the ETH. Indeed, having fixed $\varepsilon$ we can:
\begin{itemize}
\item take an instance of \tcnfsat and using Sparsification Lemma in $\Ohstar(2^{\varepsilon n/2 })$ time express it as a disjunction of at most $2^{\varepsilon n/2}$ \tcnfsat instances, each containing at most $Cn$ clauses for some constant $C$;
\item reduce each instance in polynomial time via \setsplitting and \adp to \lce, thus obtaining at most $2^{\varepsilon n/2}$ instances of \lce, each having $|V|,|E^+|,|E^-|\leq C'n$ for some constant $C'$;
\item in each of the instances run the assumed algorithm for \lce, running in $O(2^{\delta(|V|+|E^+|+|E^-|)})$ time, for $\delta=\frac{\varepsilon}{6C'}$.
\end{itemize}

\end{proof}

\subsection{A single-exponential algorithm for \lce}

Note that the trivial brute-force algorithm for \lce checks all possible orderings, working in $\Ohstar(n!)$ time.
To complete the picture of the complexity of \lce, we show that a simple dynamic programming approach can give single-exponential time complexity.
This matches the lower bound obtained from under Exponential Time Hypothesis (up to a constant
in the base of the exponent).

Before we proceed with the description of the algorithm, let us state a combinatorial observation that will be its main ingredient. Let $(V,E^+,E^-)$ be the given \lce instance. For $X\subseteq V$ and $v\notin X$ we will say that $v$ is {\em good} for the set $X$ iff
\begin{itemize}
\item no vertex $w\in X$ that is adjacent to $v$ via a negative edge is simultaneously adjacent to some vertex from $V\setminus (X\cup\{v\})$ via a positive edge;
\item no vertex $w\in V\setminus (X\cup\{v\})$ that is adjacent to $v$ via a negative edge is simultaneously adjacent to some vertex from $X$ via a positive edge.
\end{itemize}

\begin{lemma}\label{lem:good}
An ordering $\pi$ is a feasible solution of $(V,E^+,E^-)$ if and only if every vertex $v\in V$ is good for the set $\{u:u<_{\pi}v\}$.
\end{lemma}
\begin{proof}
One direction is obvious: if $\pi$ is a feasible solution, then every vertex $v$ has to be good for the set $\{u:u<_{\pi}v\}$.
If $v$ would not be good for $\{u:u<_{\pi}v\}$, there would exist a vertex $w$ certifying that $v$ is not good,
and the condition imposed upon $w$ would be not satisfied.

Now assume that every vertex $v\in V$ is good for $\{u:u<_{\pi}v\}$ and take an arbitrary vertex $v\in V$.
If there were vertices $u_1<_{\pi}u_2<_{\pi}v$ such that $u_1v\in E^+$ while $u_2v\in E^-$,
then $u_2$ would not be good for the set $\{u:u<_{\pi}u_2\}$, a contradiction.
Similarly, if there were vertices $u_1>_{\pi}u_2>_{\pi}v$ such that $u_1v\in E^+$ while $u_2v\in E^-$,
then $u_2$ would not be good for the set $\{u:u<_{\pi}u_2\}$, a contradiction as well.
Therefore, the condition imposed on $v$ is satisfied for an arbitrary choice of $v$. 
\end{proof}

We are now ready to provide the details of the algorithm.

\begin{theorem}
\lce can be solved in $\Ohstar(2^n)$ time and space complexity. Moreover, the algorithm can also output a feasible ordering of the vertices, if it exists.
\end{theorem}
\begin{proof}
Let $(V,E^+,E^-)$ be the given \lce instance. Let $W=\{(v,X): v \hbox{ is good for }X\}$. Let us construct a directed graph $D=(W,F)$, where $((v,X),(v',X'))\in F$ if and only if $X'=X\cup\{v\}$. As recognizing being good is clearly a polynomial time operation, the graph $D$ can be constructed in $\Ohstar(2^n)$ time and has that many vertices and edges. Observe that by Lemma~\ref{lem:good} there is a feasible ordering $\pi$ if and only if some sink $(v,V\setminus\{v\})$ is reachable from some source $(u,\emptyset)$; indeed, such a path corresponds to introducing the vertices of $V$ one by one in such a manner that each of them is good for the respective prefix. Reachability of any sink from any source can be, however, tested in time linear in the size of the graph using a breadth-first search. The search can also reconstruct the path in the same complexity, thus constructing the feasible solution. 
\end{proof}

\section{Conclusions}\label{sec:conc}

In this paper we addressed a number of problems raised by Kermarrec and Thraves in~\cite{francuzi} for embeddability of a signed graph into a line. We refined their study of the case of a complete signed graph by showing relation with proper interval graphs. Moreover, we have proven $NP$-hardness of the general case and shown an almost complete picture of its complexity.

Although the general case of the problem appears to be hard, real-life social networks have a certain structure.
Is it possible to develop faster, maybe even polynomial-time algorithms for classes of graphs reflecting this structure?
Can we make use of good combinatorial or spectral behaviour of real-life instances?

\bibliographystyle{plain}
\bibliography{line-cluster-embedding}

\end{document}